\documentclass{article}

\usepackage{amsthm}
\usepackage{amsmath}
\usepackage{amsfonts}
\usepackage{amssymb}
\usepackage{url}
\usepackage{hyperref}
\usepackage{comment}

\newtheorem{theorem}{Theorem}
\newtheorem{lemma}[theorem]{Lemma}
\newtheorem{definition}[theorem]{Definition}
\newtheorem{proposition}[theorem]{Proposition}
\newtheorem{corollary}[theorem]{Corollary}

\DeclareMathOperator*{\ho}{\overset{\perp}{\mathcal{\oplus}}}

\begin{document}
\sloppy
\title{Effective Operators in the Mathematical Theory of Composite Materials: the Hilbert Space Framework}
\title{Effective Operators in the Theory of Composites: Hilbert Space Framework}
\author{Aaron Welters\\Department of Mathematics and Systems Engineering\\ Florida Institute of Technology\\ Melbourne, FL USA}

%
%
\maketitle
%
\abstract{In this chapter, the Hilbert space framework in the mathematical theory of composite materials is introduced for studying the properties of effective operators. The goal is to introduce some of the key concepts and fundamental theorems in this area while showing that they follow naturally from using only basic results in operator theory on Hilbert spaces. These concepts include the $Z$-problem as an abstraction of a constitutive equation defined in terms of a bounded linear operator on a Hilbert space with a Hodge decomposition, direct and dual $Z$-problems with the duality interpretation of the inverse of an effective operator, and the notion of an $n$-phase composite with orthogonal $Z(n)$-subspace collection. These theorems include sufficient conditions for the existence and uniqueness of both the solution of a $Z$-problem and the effective operator of a $Z$-problem, a representation formula for the effective operator as an operator Schur complement, the Dirichlet and Thomson minimization principles for the effective operator, the result on monotonicity and concavity of the effective operator map, and the Keller-Dykhne-Mendelson duality relations. Moreover, another important theorem given here (which may also be of independent interest to systems theorists) says that an effective operator of an $n$-phase composite with orthogonal $Z(n)$-subspace collection is the Schur complement of a normalized homogeneous semidefinite operator pencil (in particular, has a Bessmertny\u{\i} realization) and, up to a unitary equivalence, the converse is also true. Finally, the general theory presented here is shown to recover classical results dealing with effective conductivity but can also be applied to many other important problems involving composites in physics and engineering, e.g., in elasticity and electromagnetism.}

\section{Introduction}

The mathematical theory of composite materials is an active area of research with a long history and several books on the subject, e.g., \cite{Cherkaev:2000:VMS}, \cite{Milton:2016:ETC}, \cite{Milton:2022:TOC}, \cite{ Grabovsky:2025:CMM}. This chapter is an introduction to the subject which focuses on the operator theory (on Hilbert spaces) aspects of some of the key concepts and results from the viewpoint of effective media. This theory provides a way to describe the aggregate behavior or effective properties of a complex, multicomponent system involving a constitutive equation or relation, in terms of a simpler system with an effective constitutive equation that defines the effective parameter, moduli, coefficient, tensor, matrix, operator, etc.

A quintessential example of this, which is discussed in more detail later in the chapter, is the continuum electrical conductivity problem with a conductivity tensor field $\sigma=\sigma(x)$ which is a periodic function of the spatial variable $x$ in a periodic medium with unit cell $\Omega\subseteq \mathbb{R}^d$ ($d=2$ or $d=3$). The constitutive equation governing the relation between the periodic electric field $E=E(x)$ and the periodic current density $J=J(x)$ is Ohm's law, i.e., $J=\sigma E$, and the effective conductivity tensor $\sigma_*$ is defined by the equation
\begin{gather}
    \langle J \rangle = \sigma_*\langle E \rangle,
\end{gather}
where $\langle \cdot \rangle$ denotes the spatial average over the unit cell $\Omega$. 
As an alternative perspective, suppose one wants to solve from Ohm's law both $E$ and $J$ given the average electric field $E_0=\langle E \rangle$, i.e, solve the periodic cell problem:
\begin{gather}
    J=\sigma E,\; \langle E \rangle=E_0.
\end{gather}
The starting point of the Hilbert space framework for this problem is to realize that it is defined on the Hilbert space $\mathcal{H}=\left[  L_{\#}^{2}\left(\Omega\right)\right]^{d}$ of periodic square-integrable vector-valued functions has an orthogonal triple decomposition of subspaces $\mathcal{H}=\mathcal{U}\ho\mathcal{E}\ho\mathcal{J}$ (a Hodge decomposition), see \eqref{HodgeDecompPeriodContinuumCond}-\eqref{HodgeDecompPeriodContinuumCond3}, that is naturally compatible to solving the cell problem. This then provides an example of a $Z$-problem, which is shown below, under some additional assumptions on the conductivity tensor $\sigma$, to be uniquely solvable and yields an equivalent way to define the effective operator $\sigma_*$. Then this leads to further studies of the properties of the effective conductivity (tensor) map
\begin{gather}
    \sigma=\sigma(x)\mapsto \sigma_*=[\sigma(\cdot)]_*.
\end{gather}

A special case of interest is when the conductivity $\sigma=\sigma(x)$ is piecewise constant in the unit cell $\Omega$ and made up of a finite number $n$ of scalar conductivities--an example of an isotropic $n$-phase composite. Equivalently, the periodic medium can decomposed into $n$ disjoint regions with characteristic functions $\chi_i$ with scalar conductivity $z_i$ in that region, for $i=1,\ldots,n$, and hence
\begin{gather}
    \sigma(x,z)=z_1\chi_1(x)+\ldots, z_n\chi_n(x),\;x\in \mathbb{R}^d,\;z=(z_1,\ldots, z_n)\in \mathbb{C}^n.
\end{gather}
This leads to study the effective conductivity (tensor) map as a function of $z$:
\begin{gather}
    z\mapsto \sigma_*(z)=[\sigma(\cdot,z)]_*,\label{defEffCondTensorMap}
\end{gather}
for $\chi_i,$ $i=1,\ldots, n$ fixed. 

There is an alternative operator perspective considered in the chapter. Treat the scalar functions $\sigma(z)=\sigma(\cdot, z)$, $\chi_i=\chi_i(\cdot)$ of $x$ as left-multiplication operators $L(z)=L_{\sigma(z)}, \Lambda_i=L_{\chi_i}$ on the Hilbert space $\mathcal{H}=\left[  L_{\#}^{2}\left(\Omega\right)\right]^{d}$ so that $L(z)F=\sigma(z)F$, $\Lambda_iF=\chi_iF, \forall F\in \left[  L_{\#}^{2}\left(\Omega\right)\right]^{d}$. Then in the space $\mathcal{L}(\mathcal{H})$ of all bounded linear operators on $\mathcal{H}$, the operators $\Lambda_i, i=1,\ldots,n$ form a resolution of the identity \begin{gather}
    \Lambda_i\in \mathcal{L}(\mathcal{H}),\; \Lambda_i^2=\Lambda_i=\Lambda_i^*,\; \Lambda_i\Lambda_j=0,\;\; i,j=1,\ldots, n,\; i\not=j,\\
    \Lambda_1+\cdots + \Lambda_n=I,
\end{gather}
and the multiplication operator $L(z)$ for conductivity $\sigma(z)$ becomes the linear operator pencil
\begin{gather}
    L(z)=z_1\Lambda_1+\cdots +z_n \Lambda_n\in\mathcal{L}(\mathcal{H}),\; z=(z_1,\ldots, z_n)\in \mathbb{C}^n.
\end{gather}
This leads to an example of a $Z(n)$-subspace collection (see Def.\ \ref{def:SubspCollMultiphasComposite}) for the $Z$-problem $(\mathcal{H}, \mathcal{U}, \mathcal{E}, \mathcal{J}, L(z))$ with effective (conductivity) operator $L_*(z)=[L(z)]_*$ (see Def.\ \ref{DefZProbMain}) associated to the orthogonal decomposition
\begin{gather}
    \mathcal{H}=\mathcal{U}\overset{\bot }{\oplus }\mathcal{E}\overset{\bot }{\oplus }\mathcal{J}=\mathcal{P}_1\overset{\bot }{\oplus }\cdots\overset{\bot }{\oplus }\mathcal{P}_n,\\
    \mathcal{P}_i=\Lambda_i\mathcal{H}=\{\chi_iF:F\in \mathcal{H}\},\; i=1,\ldots, n,
\end{gather}
in which $\Lambda_i$ is the orthogonal projection of $\mathcal{H}$ onto the closed subspace $\mathcal{P}_i$, for all $i=1,\ldots, n$.

We next provide a flavor of some results presented in this chapter. First, the map \eqref{defEffCondTensorMap} is well-defined for all $z\in (0,\infty)^n$ [in fact, on a much larger region $D$, see \eqref{defDomainDRotatedPolyplanes}] mapping into $d\times d$ matrices. Second, it is positive definite and monotonic, i.e.,
\begin{gather}
    0<\sigma_*(z)\leq \sigma_*(w),\; \text{whenever }z_i\leq w_i,\; i=1,\ldots, n,
\end{gather}
and the classical Wiener bounds hold:
\begin{gather}
    \left(\sum_{i=1}^nz_i^{-1}f_i\right)^{-1}I\leq \sigma_*(z)\leq\left(\sum_{i=1}^nz_if_i\right) I,
\end{gather}
where $I$ denotes the $d\times d$ identity matrix of $\mathbb{R}^{d\times d}$ and $f_i=\langle\chi_i\rangle$ is the volume fraction of the $i$th phase ($i=1,\ldots, n$). Here, and throughout the chapter, $\leq$ denotes the Loewner order on positive semidefinite matrices (resp.\ operators). Third, in two-dimensions ($d=2$), the effective conductivity tensor $\sigma_*(z)$ satisfies the Keller-Dykhne-Mendelson duality relation: 
    \begin{gather}
        \sigma_*(z)=R[\sigma_*(z^{-1})]^{-1}R^{-1},\;\;R=\begin{bmatrix}
            0 & 1\\
            -1 & 0
        \end{bmatrix},
    \end{gather}
where $z^{-1}=(z_1^{-1},\ldots, z_n^{-1})$. This latter relation is a consequence of the facts that $R^*=R^{-1}=-R$ and $R\mathcal{U}\subseteq\mathcal{U}$ along with the general result in two dimensions that a curl-free field when rotated by $90$ degrees produces a divergence-free field and vice versa (see \cite[p.\ 105]{Cherkaev:2000:VMS}, \cite[Sec.\ 3.1]{Milton:2022:TOC}) which, in particular, leads to the conclusion that $R\mathcal{E}\subseteq\mathcal{J}$ and $R\mathcal{J}\subseteq\mathcal{E}$.

The abstract theory of composites discussed below is an approach that generalizes this example of the Hilbert space framework based on an abstract Hodge decomposition, the $Z$-problem defined by an abstract constitutive equation (the cell problem), the effective operator defined by the $Z$-problem including the effective operator of an abstract $n$-phase composite with an orthogonal $Z(n)$-subspace collection. The purpose is make the theory and results applicable to a much broader range of problems such as in elasticity or electromagnetism as well as to other areas of science. For more on all this, including the connection to homogenization theory and open problems, see, e.g., \cite{Milton:2016:ETC}, \cite{Milton:2021:SOP}, \cite{Milton:2022:TOC}, \cite{Beard:2023:EOT}, \cite{Grabovsky:2025:CMM} and references therein.

In this chapter, the notation $\mathcal{H}$ will be used to denote a Hilbert space over either the real field $\mathbb{R}$ or the complex field $\mathbb{C}$, i.e., a real or complex Hilbert space, respectively. The associated inner product is denoted by $\langle \cdot, \cdot \rangle$ (or $\langle \cdot, \cdot \rangle_{\mathcal{H}}$) and it will be assumed in the complex case that $\langle u,\lambda v \rangle=\lambda\langle u,v \rangle$ for all $\lambda \in \mathbb{C}, u,v\in \mathcal{H}$. The space of all bounded linear operators that maps one Hilbert space $\mathcal{H}$ into another $\mathcal{V}$ will be denoted by $\mathcal{L}(\mathcal{H},\mathcal{V})$ , and $\mathcal{L}(\mathcal{H},\mathcal{H})$ is also denoted by $\mathcal{L}(\mathcal{H})$. For any Hilbert space $\mathcal{H}$ with an orthogonal subspace decomposition $\mathcal{H}=\mathcal{H}_0\overset{\bot
			}{\mathcal{\oplus}}\mathcal{H}_1$, recall that each operator $A\in \mathcal{L}(\mathcal{H})$ can be written uniquely as a $2\times 2$ block operator matrix with respect to this decomposition, i.e.,
            \begin{gather*}
                A=[A_{ij}]_{i,j=1,2}=\begin{bmatrix}
                    A_{11} & A_{12}\\
                    A_{21} & A_{22}
                \end{bmatrix},
            \end{gather*}
whose action on elements of $\mathcal{H}_0\overset{\bot
			}{\mathcal{\oplus}}\mathcal{H}_1$, treated as $2\times 1$ block column vectors, is
            \begin{gather*}
                A\begin{bmatrix}
                    x_0\\
                    x_1
                \end{bmatrix}=\begin{bmatrix}
                    A_{00}x_0+A_{01}x_1\\
                    A_{10}x_0+A_{11}x_1
                \end{bmatrix},\; x_0\in \mathcal{H}_0,\; x_1\in \mathcal{H}_1.
            \end{gather*}
Moreover, in the case that the $(1,1)$-block $A_{11}$ is invertible, the Schur complement of $A$ with respect to $A_{11}$, denoted by $A/A_{11}$, is defined by
\begin{gather*}
    A/A_{11}=A_{00}-A_{01}A_{11}^{-1}A_{10}.
\end{gather*}

\section{Abstract Theory of Composites}

\begin{definition}[$Z$-problem and effective operator]\label{DefZProbMain}
	The $Z$-problem
	\begin{equation}
		(\mathcal{H},\mathcal{U},\mathcal{E},\mathcal{J},L), \label{DefZProb}
	\end{equation}
	is the following problem associated with a Hilbert space $\mathcal{H}$, an orthogonal triple decomposition of $\mathcal{H}$ as 
	\begin{equation}
		\mathcal{H=U}\overset{\bot}{\mathcal{\oplus}}\mathcal{E}\overset{\bot
			}{\mathcal{\oplus}}\mathcal{J},\label{DefZProbHOrthTri}
	\end{equation}
	and a bounded linear operator $L\in\mathcal{L}(\mathcal{H})$:
	given $E_{0}\in\mathcal{U}$, find triples $\left(J_{0},E,J\right)\in\mathcal{U}\times\mathcal{E}\times\mathcal{J}$ satisfying 
	\begin{equation}
		J_{0}+J=L \left(  E_{0}+E\right).
		\label{DefZProbEq} 
	\end{equation}
	Such a triple $\left(  J_{0},E,J\right)$ is called a solution of the $Z$-problem at $E_{0}$.
	If there exists a bounded linear operator $L_*\in\mathcal{L}(\mathcal{U}$) such that 
	\begin{equation}
		J_{0}=L_{\ast}E_{0}, \label{DefZProbEffOp}
	\end{equation}
	whenever $E_0\in\mathcal{U}$ and $\left(  J_{0},E,J\right)$ is a solution of the $Z$-problem at $E_0$, then $L_*$ is called an effective operator of the $Z$-problem.
\end{definition}

Now let $(\mathcal{H},\mathcal{U},\mathcal{E},\mathcal{J},L)$ be a $Z$-problem. Then one represents the operator
\begin{align}
 L=[L_{ij}]_{i,j=0,1,2}=\begin{bmatrix}
     L_{11} & L_{12} & L_{13}\\
     L_{21} & L_{22} & L_{23}\\
     L_{31} & L_{32} & L_{33}
 \end{bmatrix}\in\mathcal{L}(\mathcal{H}) \label{3b3BlockOpReprOfSigma}  
\end{align}
as a $3\times 3$ block operator matrix with respect to the orthogonal triple decomposition (\ref{DefZProbHOrthTri}) of the Hilbert space $\mathcal{H}=\mathcal{U}\ho\mathcal{E}\ho\mathcal{J}$. More precisely, introducing the orthogonal projections $\Gamma_0,\Gamma_1,\Gamma_2$ of $\mathcal{H}$ onto $H_0=\mathcal{U}, H_1=\mathcal{E}, H_2=\mathcal{J},$ respectively, the operators $L_{ij}$ are defined by
\begin{align}
    L_{ij}\in \mathcal{L}(H_j,H_i),\;L_{ij}=\Gamma_iL\Gamma_j:H_j\rightarrow H_i,\label{DefOfSigmaSubblocks}
\end{align}
for $i,j=0,1,2$. In particular, $L_{11}$ is the compression of $L$ to $\mathcal{E}$, that is,
\begin{align}
   L_{11}= \Gamma_1L\Gamma_1|_{\mathcal{E}},\label{DefAltSigma11Subblock}
\end{align}
i.e., the restriction of the operator $\Gamma_1L\Gamma_1$ on $\mathcal{H}$ to the closed subspace $\mathcal{E}$. The $Z$-problem (\ref{DefZProbEq}) is then equivalent to the system
\begin{gather}
    L_{00}E_0+L_{01}E=J_0,\label{ZProbEquivFormPart1}\\
        L_{10}E_0+L_{11}E=0,\label{ZProbEquivFormPart2}\\
        L_{20}E_0+L_{21}E=J.\label{ZProbEquivFormPart3}
\end{gather}
In what follows, the following hypotheses will appear:

\begin{description}
    \item[\hypertarget{(H0)}{(H0)}] $L_{11}$ is invertible.
    \item[\hypertarget{(H1)}{(H1)}] $L^*=L$, $L_{11}\geq 0$, $L_{11}$ is invertible.
    \item[\hypertarget{(H2)}{(H2)}] $L^*=L\geq 0$, $L$ is invertible.
\end{description}
These statements have been ordered from weakest to strongest, in the sense that
\begin{gather}
    \hyperlink{(H2)}{(H2)}\Rightarrow\hyperlink{(H1)}{(H1)}\Rightarrow\hyperlink{(H0)}{(H0)}.
\end{gather}
Another common hypothesis on $L$ is the coercivity assumption:
\begin{description}
     \item[\hypertarget{(LM)}{(LM)}]There exist scalars $\lambda\in \mathbb{C}\setminus\{0\}$ and $\delta>0$ such that 
     \begin{gather}
         \operatorname{Re}(\lambda L)\geq \delta I,
     \end{gather}
\end{description}
where $I$ is the identity operator and $\operatorname{Re}(\cdot)$ denotes the self-adjoint part of an operator $(\cdot)$ in $\mathcal{L}(\mathcal{H})$, i.e., 
\begin{gather}
    \operatorname{Re}A=\frac{1}{2}(A+A^*), \;A\in\mathcal{L}(\mathcal{H}).
\end{gather}
By the Lax-Milgram lemma \cite[p.\ 57, Theorem 6]{Lax:2002:FA}, any operator $L$ satisfying \hyperlink{(LM)}{(LM)} is invertible and hence, using similar reasoning, so are $L_{ii}$, $i=0,1,2$. In particular, this shows that
\begin{gather}
    \hyperlink{(LM)}{(LM)}\Rightarrow \hyperlink{(H0)}{(H0)}.\label{LMImpliesH0}
\end{gather}

In any case, one gets the classical formulas for the solution of the $Z$-problem and the effective operator as a Schur complement (see \cite[Sec.\ 12.7, eq.\ (12.57)]{Milton:2022:TOC}, \cite[Sec.\ 10.2]{Cassier:2016:RAF}, and \cite[Sec.\ II]{Beard:2023:EOT}):
\begin{gather}
    J_0=L_*E_0,\;E=-L_{11}^{-1}L_{10}E_0,\; J=L_{20}E_0+L_{21}E,\label{ClassicSolnZProb}\\
    L_*=\begin{bmatrix}
        L_{00}&L_{01}\\
        L_{10}&L_{11}
    \end{bmatrix}/L_{11}=L_{00}-L_{01}L_{11}^{-1}L_{10},\label{ClassicEffOperFormula}
\end{gather}
where the latter is the Schur complement with respect to $L_{11}$ of the compression of $L$ to $\mathcal{U}\overset{\bot}{\mathcal{\oplus}}\mathcal{E}$:
\begin{gather}
    \begin{bmatrix}
        L_{00}&L_{01}\\
        L_{10}&L_{11}
    \end{bmatrix}=(\Gamma_0+\Gamma_1)L(\Gamma_0+\Gamma_1)|_{\mathcal{U}\overset{\bot}{\mathcal{\oplus}}\mathcal{E}}\in \mathcal{L}(\mathcal{U}\overset{\bot}{\mathcal{\oplus}}\mathcal{E}).
\end{gather}
These results are summarized by the following theorem.
\begin{theorem}\label{ThmMainClassicalZProbEffOp}
If $(\mathcal{H},\mathcal{U},\mathcal{E},\mathcal{J},L)$ is a $Z$-problem and $L_{11}$ is invertible [i.e., \hyperlink{(H0)}{(H0)} holds] then the $Z$-problem has a unique solution for each $E_0\in \mathcal{U}$ and it is given by the formulas (\ref{ClassicSolnZProb}), (\ref{ClassicEffOperFormula}). Moreover, the effective operator of the $Z$-problem exists, is unique, and is given by the Schur complement formula (\ref{ClassicEffOperFormula}).
\end{theorem}

This next definition introduces the notation of duality in $Z$-problems.
\begin{definition}
[Dual $Z$-problem]\label{DefDualZProb}Given a $Z$-problem $(\mathcal{H},\mathcal{U},\mathcal{E},\mathcal{J},L)$ with invertible $L$, called the direct $Z$-problem, the corresponding dual $Z$-problem is the $Z$-problem $(\mathcal{H},\mathcal{U},\mathcal{J},\mathcal{E},L^{-1})$. An effective operator of the dual $Z$-problem will be denoted by $\left(
L^{-1}\right)_{*'}$.
\end{definition}

Now, Theorem \ref{ThmMainClassicalZProbEffOp} yields immediately the next fundamental result on duality. 
\begin{corollary}\label{cor:FundThmOnDualityZProbs}
    If $(\mathcal{H},\mathcal{U},\mathcal{E},\mathcal{J},L)$ is a $Z$-problem such that hypothesis \hyperlink{(LM)}{(LM)} is true for $L$, then $L$ is invertible and \hyperlink{(LM)}{(LM)} is also true for $L^{-1}$. Moreover, the hypothesis \hyperlink{(LM)}{(LM)} is true for the effective operator $L_*$ of the direct $Z$-problem, it is invertible, and the effective operator $\left(
L^{-1}\right)_{*'}$ of the dual $Z$-problem $(\mathcal{H},\mathcal{U},\mathcal{J},\mathcal{E},L^{-1})$ is given by the duality formula
\begin{gather}
    \left(
L^{-1}\right)_{*'}=(L_*)^{-1}.
\end{gather}
\end{corollary}

One immediate application of this result is the following abstract version of the Keller-Dykhne-Mendelson duality \cite[Chap.\ 3]{Milton:2022:TOC}.

\begin{corollary}\label{cor:AbstractKellerDykhneMendelsonDuality}
    If $\mathcal{H}$ is a Hilbert space with orthogonal triple decomposition \eqref{DefZProbHOrthTri} and $R\in \mathcal{L}(\mathcal{H})$ is invertible with $R^*=R^{-1}=-R,$  $R\mathcal{U}\subseteq \mathcal{U}$, $R\mathcal{E}\subseteq \mathcal{J}$, and $R\mathcal{J}\subseteq \mathcal{E}$, then $R\mathcal{U}=\mathcal{U}, R\mathcal{E}=\mathcal{J}$, and $R\mathcal{J}=\mathcal{E}$. Moreover, for any $L\in \mathcal{L}(\mathcal{H})$ in which hypothesis \hyperlink{(LM)}{(LM)} is true, the following equality holds (as functions on $\mathcal{U}$):
    \begin{gather}
        L_*=R[(R^{-1}L^{-1}R)_*]^{-1}R^{-1}.
    \end{gather}
\end{corollary}
\begin{proof}
    First, the assumptions imply that $R^{-1}\mathcal{U}=R\mathcal{U}\subseteq \mathcal{U}$, which implies $R\mathcal{U}=\mathcal{U}$. Similarly, $R^{-1}\mathcal{E}=R\mathcal{E}\subseteq\mathcal{J}$ and $R\mathcal{J}\subseteq \mathcal{E}$ imply $R\mathcal{E}=\mathcal{J}$ and $\mathcal{E}=R\mathcal{J}$. Now suppose that $L\in \mathcal{L}(\mathcal{H})$ satisfies the hypothesis \hyperlink{(LM)}{(LM)}. Then \hyperlink{(LM)}{(LM)} also holds for $R^{-1}LR$. By Theorem \ref{ThmMainClassicalZProbEffOp} and Corollary \ref{cor:FundThmOnDualityZProbs}, the $Z$-problems $(\mathcal{H},\mathcal{U},\mathcal{E},\mathcal{J},L)$ and $(\mathcal{H},\mathcal{U},\mathcal{E},\mathcal{J},R^{-1}L^{-1}R)$ have unique effective operators $L_*$ and $(R^{-1}LR)_*$, respectively, such that $(R^{-1}LR)_*$ is invertible. Let $E_0\in \mathcal{U}$. Then there exists $\left(J_{0},E,J\right)\in\mathcal{U}\times\mathcal{E}\times\mathcal{J}$ satisfying 
	\begin{equation}
		J_{0}+J=L \left(  E_{0}+E\right),
	\end{equation}
    and hence $J_0=L_*E_0$. But then
    \begin{equation}
		(R^{-1}L^{-1}R)(R^{-1}J_{0}+R^{-1}J)=R^{-1}E_{0}+R^{-1}E
	\end{equation}
    with $R^{-1}J_{0}\in \mathcal{U}$ and $\left(R^{-1}E_{0},R^{-1}J,R^{-1}E\right)\in\mathcal{U}\times\mathcal{E}\times\mathcal{J}$, which implies $R^{-1}E_{0}=(R^{-1}L^{-1}R)_*R^{-1}J_{0}$. Thus, one has $R[(R^{-1}L^{-1}R)_*]^{-1}R^{-1}E_0=J_0=L_*E_0$. This proves the corollary. 
\end{proof}

Under hypothesis \hyperlink{(H1)}{(H1)}, the following theorem gives a variational formulation for the effective operator \cite[Chap.\ 13]{Milton:2022:TOC}, \cite[Sec.\ 2.7]{Milton:2016:ETC}, \cite{Beard:2023:EOT}.
\begin{theorem}[Dirichlet minimization principle]\label{ThmClassicalDiriMinPrin}
If $L\in \mathcal{L}(\mathcal{H}),$ $L^*=L$, $L_{11}\geq 0$, and $L_{11}$ is invertible [i.e., \hyperlink{(H1)}{(H1)} holds] then the effective operator $L_*$ is the unique self-adjoint operator satisfying the minimization principle
\begin{equation}
		( E_0,L_*E_0 )=\min_{E\in\mathcal{E}}( E_0+E,L(E_0+E) ),\;\forall E_0\in\mathcal{U},\label{ClassicalDirMinPrincEffOp}
	\end{equation}
	and, for each $E_0\in\mathcal{U}$, the minimizer is unique and given by
	\begin{equation}
		E=-L_{11}^{-1}L_{10}E_0.
	\end{equation}
	Moreover, one has the following upper bound on the effective operator:
	\begin{equation}
		L_*\leq L_{00},
	\end{equation}
	where $L_{00}$ is the compression of $L$ to $\mathcal{U}$, i.e.,
	\begin{align}
	    L_{00}=\Gamma_0L\Gamma_0|_{\mathcal{U}}.
	\end{align}
\end{theorem}

This theorem has several immediate consequences. One is the Thomson minimization principle \cite[Chap.\ 13]{Milton:2022:TOC}, \cite[Sec.\ 2.7]{Milton:2016:ETC}, \cite{Beard:2023:EOT} which is an immediate consequence of duality, using Theorem \ref{ThmClassicalDiriMinPrin} and Corollary \ref{cor:FundThmOnDualityZProbs}.
\begin{corollary}[Thomson minimization principle]\label{ThmClassicalThomMinPrin}
If $L\in \mathcal{L}(\mathcal{H}),$ $L^*=L\geq 0,$ and $L$ is invertible [i.e., \hyperlink{(H2)}{(H2)} holds], then $L_*$ is invertible and $(L_*)^{-1}$ is the unique self-adjoint operator satisfying the minimization principle
	\begin{gather}
		(J_0,(L_*)^{-1}J_0 )=\min_{J\in\mathcal{J}}( J_0+J,L^{-1}(J_0+J)),\; \forall J_0\in\mathcal{U},
	\end{gather}
	and, for each $J_0\in\mathcal{U}$, the minimizer is unique and given by
	\begin{equation}
		J=-L_{22}^{-1}L_{20}J_0.
	\end{equation}
	Moreover, one has the following upper and lower bounds on the effective operator:
	\begin{equation}
		0\leq [(L^{-1})_{00}]^{-1}\leq L_*\leq L_{00},\label{ClassicalUpperLowerBddsEffOp}
	\end{equation}
	where $(L^{-1})_{00}$ is the compression of $L^{-1}$ to $\mathcal{U}$, i.e.,
	\begin{align}
	    (L^{-1})_{00}=\Gamma_0L^{-1}\Gamma_0|_{\mathcal{U}},
	\end{align}
    which is invertible in $\mathcal{L}(\mathcal{U})$ with inverse $[(L^{-1})_{00}]^{-1}$.
\end{corollary}

The next result is another consequence of Theorem \ref{ThmClassicalDiriMinPrin} and it gives the monotonicity and concavity of the effective operator map (cf.\ \cite[Sec.\ 13.2]{Milton:2022:TOC}, \cite[Cor.\ 5]{Beard:2024:MMC})
\begin{gather}
    (\cdot)_*:\mathcal{L}(\mathcal{H})_{++}\rightarrow \mathcal{L}(\mathcal{U})_{++},\;L\mapsto L_*,\;\forall L\in \mathcal{L}(\mathcal{H})_{++},
\end{gather}
where $\mathcal{L}(\mathcal{H})_{++}$ denotes the subset of $\mathcal{L}(\mathcal{H})$ of all positive operators, i.e., all $L\in \mathcal{L}(\mathcal{H})$ satisfying \hyperlink{(H2)}{(H2)}, and $\mathcal{L}(\mathcal{U})_{++}$ is defined similarly. 
\begin{corollary}\label{prop:EffOpMonotConvavity}
    Let $L, M\in \mathcal{L}(\mathcal{H})_{++}$. Then the following statements are true:
    \begin{enumerate}
        \item[(i)] If $L\leq M$, then \begin{equation}
            L_*\leq M_*.
        \end{equation}
        \item[(ii)] If $t\in [0,1]$, then
            \begin{gather}
                t L_*+(1-t)M_*\leq [tL +(1-t)M]_*.
            \end{gather}
    \end{enumerate}
\end{corollary}
\begin{proof}
    (i): If $L\leq M$, then by Theorem \ref{ThmClassicalDiriMinPrin},
    \begin{equation*}
		( E_0,L_*E_0 )=\min_{E\in\mathcal{E}}( E_0+E,L(E_0+E) )\leq \min_{E\in\mathcal{E}}( E_0+E,M(E_0+E) )=( E_0,M_*E_0 )
	\end{equation*}
    for all $E_0\in\mathcal{U}$, which proves (i).

    (ii): Fix $t\in [0,1]$. Then $\ell(t):=tL +(1-t)M\in \mathcal{L}(\mathcal{H})_{++}$ and for any $E_0\in \mathcal{U},E\in \mathcal{E}$,
    \begin{align*}
        ( E_0+E,\ell(t)(E_0+E) )&=t( E_0+E,L(E_0+E) )+(1-t)( E_0+E,M(E_0+E) )\\
        &\geq t\min_{E\in\mathcal{E}}( E_0\!+\!E,L(E_0\!+\!E) )\!+\!(1-t)\min_{E\in\mathcal{E}}( E_0\!+\!E,M(E_0\!+\!E) ).
    \end{align*}
    The proof of (ii) now follows from this by Theorem \ref{ThmClassicalDiriMinPrin}.    
\end{proof}

The next definition introduces the notion of orthogonal subspace collections and the abstraction of an isotropic multiphase composite, see \cite[Chap.\ 29]{Milton:2022:TOC}, \cite{Milton:2016:SAS}, \cite{Cassier:2016:RAF}.
\begin{definition}[Subspace collections and multiphase composites]\label{def:SubspCollMultiphasComposite}
An orthogonal $Z(n)$-subspace collection $\mathfrak{C}$ is an $(n+4)$-tuple 
\begin{gather}
    \mathfrak{C} = (\mathcal{H},\mathcal{U},\mathcal{E},\mathcal{J},\mathcal{P}_1,\ldots,\mathcal{P}_n),\label{defSubspaceCollection}
\end{gather}
where $\mathcal{H}$ is a Hilbert space with subspaces $\mathcal{U},\mathcal{E},\mathcal{J},\mathcal{P}_1,\ldots,\mathcal{P}_n$ satisfying
\begin{gather}
    \mathcal{H}=\mathcal{U}\overset{\bot }{\oplus }\mathcal{E}\overset{\bot }{\oplus }\mathcal{J}=\mathcal{P}_1\overset{\bot }{\oplus }\cdots\overset{\bot }{\oplus }\mathcal{P}_n.
\end{gather}
Let $\Lambda_i$ denote the orthogonal projection of $\mathcal{H}$ onto $\mathcal{P}_i$, $i=1,\ldots, n$. Then the $Z$-problem 
\begin{gather}
    (\mathcal{H},\mathcal{U},\mathcal{E},\mathcal{J}, L(z)),\label{DefZnProb1}
\end{gather}
where
\begin{gather}
    L(z)=z_1\Lambda_1+\ldots+z_n \Lambda_n,,\; z=(z_1,\ldots, z_n)\in\mathbb{C}^n\label{DefZnProb2}
\end{gather}
is called an (abstract) $n$-phase composite with (orthogonal) $Z(n)$-subspace collection $\mathfrak{C}$.
\end{definition}
Now one is interested in the properties of the associated effective operator map:
\begin{gather}
    L_*(\cdot):D\rightarrow \mathcal{L}(\mathcal{U}),\;
    z\mapsto L_*(z),\\
    D=\bigcup_{\theta\in [0,2\pi)}(e^{i\theta}\mathbb{C}^{+})^n,\label{defDomainDRotatedPolyplanes}
\end{gather}
where $\mathbb{C}^{+}$ denotes the open upper-half of the complex plane $\mathbb{C}$, i.e.,
\begin{gather}
    \mathbb{C}^{+}=\{z\in\mathbb{C}:\operatorname{Im}(z)>0\}.
\end{gather}
The next result from \cite[Sec.\ 3.4, Prop.\ 24]{Stefan:2021:SCA} shows this map is a well-defined function on the domain $D\subseteq \mathbb{C}^n$ with a representation of a special type called a Bessmertny\u{\i} realization.

\begin{proposition}[Effective operator: Bessmertny\u{\i} realization]\label{prop:EffOpBessRealization}
If $z\in D$, then the $Z$-problem \eqref{DefZnProb1}, \eqref{DefZnProb2} has a unique solution for any $E_0\in\mathcal{U}$ and it is given by the formulas (\ref{ClassicSolnZProb}), (\ref{ClassicEffOperFormula}) with $L=L(z)=[L_{ij}(z)]_{i,j=0,1,2}$. Moreover, the effective operator $L_*(z)$ of the $Z$-problem exists, is unique, and is given as the Schur complement
\begin{gather}
    L_*(z)=A(z)/A_{11}(z)=A_{00}(z)-A_{01}(z)A_{11}(z)^{-1}A_{10}(z)
\end{gather}
of a normalized homogeneous positive semidefinite operator pencil
\begin{gather}
    A(z)=\begin{bmatrix}
  A_{00}(z) & A_{01}(z) \\
 A_{10}(z) & A_{11}(z)
\end{bmatrix},\label{EffOpBessRealization1}\\
A(z)= z_1A_1+\cdots +z_n A_n,\label{EffOpBessRealization2}\\
A_i^*=A_i\geq 0,\; (i=1,\ldots, n),\label{EffOpBessRealization3}\\
A_1+\cdots+A_n=I,\label{EffOpBessRealization4}
\end{gather}
where $I$ denotes the identity operator on $\mathcal{U}\overset{\bot}{\mathcal{\oplus}}\mathcal{E}$, and
\begin{gather}
   A(z)=[A_{ij}(z)]_{i,j=0,1}=[L_{ij}(z)]_{i,j=0,1}=(\Gamma_0+\Gamma_1)L(z)(\Gamma_0+\Gamma_1)|_{\mathcal{U}\overset{\bot}{\mathcal{\oplus}}\mathcal{E}}, \label{EffOpBessRealization5}\\
    A_i=(\Gamma_0+\Gamma_1)\Lambda_i(\Gamma_0+\Gamma_1)|_{\mathcal{U}\overset{\bot}{\mathcal{\oplus}}\mathcal{E}}\in \mathcal{L}(\mathcal{U}\overset{\bot}{\mathcal{\oplus}}\mathcal{E}),\;\;i=1,\ldots, n,\label{EffOpBessRealization6}
\end{gather}
i.e., $A(z)$ and $A_i$ are the compression of $L(z)$ and $\Lambda_i$, respectively, to $\mathcal{U}\overset{\bot}{\mathcal{\oplus}}\mathcal{E}$.
\end{proposition}
\begin{proof}
    The proof follows immediately from Theorem \ref{ThmMainClassicalZProbEffOp} using the fact that, for any $z\in D,$ the hypothesis \hyperlink{(LM)}{(LM)} holds for the operator $L(z)$ [and hence, by \eqref{LMImpliesH0}, so does hypothesis \hyperlink{(H0)}{(H0)}].
\end{proof}

This next corollary is the abstract version of the Keller-Dykhne-Mendelson duality for $n$-phase composites, cf.\ \cite[Chap.\ 3]{Milton:2022:TOC}.
\begin{corollary}\label{cor:AbstractMultiPhaseKellerDykhneMendelsonDuality}
     Suppose $(\mathcal{H},\mathcal{U},\mathcal{E},\mathcal{J}, L(z))$ is the $Z$-problem of an abstract $n$-phase composite with orthogonal $Z(n)$-subspace collection  $\mathfrak{C} = (\mathcal{H},\mathcal{U},\mathcal{E},\mathcal{J},\mathcal{P}_1,\ldots,\mathcal{P}_n)$ such that there exists $R\in \mathcal{L}(\mathcal{H})$ which is invertible with $R^*=R^{-1}=-R, R\mathcal{U}\subseteq \mathcal{U}$, $R\mathcal{E}\subseteq \mathcal{J}$, $R\mathcal{J}\subseteq \mathcal{E}$, and $R\mathcal{P}_i\subseteq \mathcal{P}_i$ for all $i=1,\ldots, n$. Then the following equality holds (as functions on $\mathcal{U}$) for all $z\in D$:
    \begin{gather}
        L_*(z)=R[L_*(z^{-1})]^{-1}R^{-1},
    \end{gather}
    where $z^{-1}=(z_1^{-1},\ldots, z_n^{-1})$.
\end{corollary}
\begin{proof}
    First, it follows from the assumptions that $R\mathcal{P}_i=\mathcal{P}_i$ so that $R\Lambda_i=\Lambda_iR$ for each $i=1,\ldots, n$. Second, if $z=(z_1,\ldots, z_n)\in \mathbb{C}^n$ such that $z_i\not=0$ for all $i$, then $L(z)^{-1}=L(z^{-1})$ and hence $R^{-1}L(z)^{-1}R=R^{-1}L(z^{-1})R=L(z^{-1})$. In particular, if $z\in D$, then $z^{-1}\in D$ and $L_*(z^{-1})=[L(z^{-1})]_*=[L(z)^{-1}]_*$. Combining this with Corollary \ref{cor:AbstractKellerDykhneMendelsonDuality}, for any $z\in D$ one has
    \begin{gather}
        L_*(z)=R[(R^{-1}L(z)^{-1}R)_*]^{-1}R^{-1}=R[L(z^{-1})]_*^{-1}R^{-1}=R[L_*(z^{-1})]^{-1}R^{-1}.
    \end{gather}

\end{proof}

This next corollary is the abstract version of the Wiener bounds and monotonicity of the effective operator for $n$-phase composites, see \cite[Chaps.\ 13 \& 22]{Milton:2022:TOC}, \cite[Chap.\ 2, Sec.\ 7]{Milton:2016:ETC}, \cite[Chap.\ 3, Sec.\ 1]{Cherkaev:2000:VMS}.
\begin{corollary}\label{cor:AbstWienerBounds}
(i) If $z_i>0$ for $i=1,\ldots,n$, then
\begin{gather}
    \left(\sum_{i=1}^nz_i^{-1}\Gamma_0\Lambda_i\Gamma_0|_{\mathcal{U}}\right)^{-1}\leq L_*(z)\leq\sum_{i=1}^nz_i\Gamma_0\Lambda_i\Gamma_0|_{\mathcal{U}},
\end{gather}
where $\Gamma_0\Lambda_i\Gamma_0|_{\mathcal{U}}$ is the compression of $\Lambda_i$ to $\mathcal{U}$.\\
(ii) If $0<z_i\leq w_i$ for $i=1,\ldots,n$, then
\begin{gather}
    L_*(z)\leq L_*(w).
\end{gather}
\end{corollary}
\begin{proof}
    The proof of (i) is immediate from 
    Corollary \ref{ThmClassicalThomMinPrin} by noting that if $z\in (0,\infty)^n$, then
    \begin{gather}
        L(z)^*=L(z)\geq 0,\\
        L(z)_{00}=\Gamma_0L(z)\Gamma_0|_{\mathcal{U}}=\sum_{i=1}^nz_i\Gamma_0\Lambda_i\Gamma_0|_{\mathcal{U}},\\
        L(z)^{-1}=\sum_{i=1}^nz_i^{-1}\Lambda_i,\;[(L(z)^{-1})_{00}]^{-1}=\left(\sum_{i=1}^nz_i^{-1}\Gamma_0\Lambda_i\Gamma_0|_{\mathcal{U}}\right)^{-1}.
    \end{gather}
Next, the proof of (ii) follows immediately from Corollary \ref{prop:EffOpMonotConvavity}, since if $0<z_i\leq w_i$ for $i=1,\ldots, n$, then $L(z),L(w)\in \mathcal{L}(\mathcal{H})_{++}$ and
\begin{gather}
    L(z)=\sum_{i=1}^nz_i\Gamma_i\leq \sum_{i=1}^nw_i\Gamma_i=L(w). 
\end{gather}
\end{proof}

Proposition \ref{prop:EffOpBessRealization} gives an indication that there is a deeper relationship between Bessmertny\u{\i} realizations and effective operators. Thus, a natural question to ask is whether the converse of Proposition \ref{prop:EffOpBessRealization} is true. The next theorem says that it is, up to a unitary transformation. 

To begin, the following lemma is needed.
\begin{lemma}\label{lem:TUnitaryTransRestrV}
    If $\mathcal{V}, \mathcal{W}$ are Hilbert spaces, $V\in \mathcal{L}(\mathcal{V}, \mathcal{W})$, and $V^*V=I$, then the following statements are true:
    \begin{itemize}
         \item[(i)] $VV^*$ is the orthogonal projection of $\mathcal{W}$ onto $V\mathcal{V}$ (the range of $V$).
         \item[(ii)]\; The restrictions of the operators $V$ and $V^*$, defined by
        \begin{gather}
            T:\mathcal{V}\rightarrow V\mathcal{V}, Tv=Vv,\; v\in \mathcal{V},\label{DefTUnitaryTransRestrV}\\ S:V\mathcal{V}\rightarrow \mathcal{V},\; Sy=V^*y, \; y\in V\mathcal{V},
        \end{gather}
        respectively, are unitary transformations that are inverses of each other, i.e.,
        \begin{gather}
            T\in \mathcal{L}(\mathcal{V},V\mathcal{V}),\;S\in \mathcal{L}(V\mathcal{V},\mathcal{V}), \; T^{-1}=T^*=S.
        \end{gather}
    \end{itemize}
\end{lemma}
\begin{proof}
    (i): As $V^*\in \mathcal{L}(\mathcal{W},\mathcal{V})$ is the Hilbert space adjoint of $V\in \mathcal{L}(\mathcal{V},\mathcal{W})$ and $V^*V=I$ (the identity operator on $\mathcal{V}$), the operator $P:=VV^*\in \mathcal{L}(\mathcal{W})$ is an orthogonal projection of $\mathcal{W}$ onto $P\mathcal{W}$, since $P^2=VV^*VV^*=VV^*=P$ and $P^*=(VV^*)^*=(V^*)^*V^*=VV^*=P$. Next, clearly $P\mathcal{W}\subseteq V\mathcal{V}=VV^*V\mathcal{V}\subseteq P\mathcal{W}$, hence $P\mathcal{W}=V\mathcal{V}$. This proves (i).

    (ii): Clearly, $S,T$ are well-defined functions such that $T\in \mathcal{L}(\mathcal{V},V\mathcal{V})$ and $S\in \mathcal{L}(V\mathcal{V},\mathcal{V})$. Next, as $V^*V=I$ and, by part (i), $VV^*$ is the orthogonal projection of $\mathcal{W}$ onto $V\mathcal{V}$, one has
    \begin{gather}
        STv=V^*Vv=v,\; \forall v\in \mathcal{V},\\
        TSy=(VV^*)y=y,\;\forall y\in V\mathcal{V}
    \end{gather}
    which proves that $T$ is invertible with $T^{-1}=S$. Finally, denoting the inner product of a Hilbert space $H$ by $\langle \cdot,\cdot\rangle_{H}$,
    \begin{gather}
        \langle Tv, y\rangle_{V\mathcal{V}}=\langle Vv, y\rangle_{\mathcal{W}}=\langle v, V^*y\rangle_{\mathcal{V}}=\langle v, Sy\rangle_{\mathcal{V}},\;\text{ for all }v\in \mathcal{V}, y\in V\mathcal{V},
    \end{gather}
    implying $T^*=S$, which proves (ii).
\end{proof}

\begin{theorem}\label{thm:BessNHPSDPCharViaEffOp}
    If $f(z)$ is the Schur complement of a normalized homogeneous positive semidefinite operator pencil, i.e., for all $z\in D$,
    \begin{gather}
        f(z)=A(z)/A_{11}(z)=A_{00}(z)-A_{01}(z)A_{11}(z)^{-1}A_{10}(z),\label{ABlockFormBessNHPSDPCharViaEffOp}
    \end{gather}
    where
    \begin{gather}
    A(z)=\begin{bmatrix}
  A_{00}(z) & A_{01}(z) \\
 A_{10}(z) & A_{11}(z)
\end{bmatrix},\\
A(z)= z_1A_1+\cdots +z_n A_n,\\
A_i^*=A_i\geq 0,\; (i=1,\ldots, n),\\
A_1+\cdots+A_n=I,\label{NormalizationCond}
\end{gather}
then there exists an effective operator $L_*(z)$ of an abstract $n$-phase composite with orthogonal Z(n)-subspace
collection and a constant unitary transformation $T$ such that
\begin{gather}
    f(z)=T^*L_*(z)T\qquad \forall z\in D.
\end{gather}
In particular, $f(\cdot)$ and $L_*(\cdot)$ are unitarily equivalent on the domain $D$.
\end{theorem}
\begin{proof}
By hypotheses, there exists a Hilbert space with orthogonal decomposition $\mathcal{H}_0\overset{\bot}{\mathcal{\oplus}}\mathcal{H}_1$ such that $A(z)=[A_{ij}(z)]_{i,j=0,1}$ is a $2\times 2$ block operator matrix \eqref{ABlockFormBessNHPSDPCharViaEffOp} with respect to that decomposition. For each $i$, there exists a Hilbert space $\mathcal{V}_i$ and a bounded linear operator $V_i\in\mathcal{L}(\mathcal{H}_0\overset{\bot}{\mathcal{\oplus}}\mathcal{H}_1,\mathcal{V}_i)$ such that
\begin{gather}
    A_i=V_i^*V_i,\; i=1,\ldots, n.
\end{gather}
(For example, $\mathcal{V}_i=\mathcal{H}_0\overset{\bot}{\mathcal{\oplus}}\mathcal{H}_1$ and $V_i=A_{i}^{1/2}$ is the positive semidefinite operator square root of $A_i$, $i=1,\ldots, n$). Next, recall that the (external) direct sum of the Hilbert spaces $\mathcal{V}_1,\ldots, \mathcal{V}_n$, i.e.,
\begin{gather}
        \mathcal{H}=\bigoplus_{i=1}^n\mathcal{V}_i=\mathcal{V}_1\oplus\cdots \oplus \mathcal{V}_n=\left\{\begin{bmatrix}
            v_1\\
            \vdots\\
            v_n
        \end{bmatrix}:v_i\in\mathcal{V}_i,i=1,\ldots, n\right\},
    \end{gather}
     is a Hilbert space 
     when equipped with the usual componentwise addition and scalar multiplication and inner product $\langle \cdot, \cdot \rangle$ defined by
     \begin{gather}
         \left\langle v,w \right\rangle = \sum_{i=1}^n\langle v_i,w_i \rangle_{\mathcal{V}_i}, \; \text{ for all }v,w\in \mathcal{H},
     \end{gather}
     where $\langle \cdot, \cdot\rangle_{\mathcal{V}_i}$ denotes the inner product of $\mathcal{V}_i$, 
     $i=1,\ldots, n$. In particular, it has the orthogonal decomposition
     \begin{gather}
         \mathcal{H}=\mathcal{P}_1\overset{\bot }{\oplus }\cdots\overset{\bot }{\oplus }\mathcal{P}_n,\\
         \mathcal{P}_i= \left\{v\in \mathcal{H}:v_j=0, j\not=i\right\},\; i=1,\ldots,n.
     \end{gather}
     Let $\Lambda_i$ denote the orthogonal projection of $\mathcal{H}$ onto $\mathcal{P}_i$, $i=1,\ldots, n$ and define
     \begin{gather}
         L(z)=z_1\Lambda_1+\cdots+z_n\Lambda_n,\;z=(z_1,\ldots, z_n)\in \mathbb{C}^n.
     \end{gather}
     Next, define
     \begin{gather}
         V\in \mathcal{L}(\mathcal{H}_0\overset{\bot}{\mathcal{\oplus}}\mathcal{H}_1,\mathcal{{H}}),\; Vx=\begin{bmatrix}
             V_1x\\
             \vdots\\
             V_nx
         \end{bmatrix},\; x\in \mathcal{H}_0\overset{\bot}{\mathcal{\oplus}}\mathcal{H}_1.
     \end{gather}
     Then the Hilbert space adjoint of $V$ is
     \begin{gather}
         V^*\in \mathcal{L}(\mathcal{{H}},\mathcal{H}_0\overset{\bot}{\mathcal{\oplus}}\mathcal{H}_1),\; V^*v=\sum_{i=1}^nV_i^*v_i,\; v\in \mathcal{H},
     \end{gather}
     since, for every $x\in \mathcal{H}_0\overset{\bot}{\mathcal{\oplus}}\mathcal{H}_1$ and $v\in \mathcal{H}$,
     \begin{gather}
        \langle Vx,v\rangle=\sum_{i=1}^n\langle V_ix,v_i \rangle_{\mathcal{V}_i}=\sum_{i=1}^n\langle x,V^*v_i \rangle_{\mathcal{H}_0\overset{\bot}{\mathcal{\oplus}}\mathcal{H}_1}=\left\langle x,\sum_{i=1}^nV_i^*v_i \right\rangle_{\mathcal{H}_0\overset{\bot}{\mathcal{\oplus}}\mathcal{H}_1}.
     \end{gather}
     It follows that
\begin{gather}
    A_i=V^*\Lambda_iV,\;i=1,\ldots, n,
\end{gather}
since
\begin{gather}
    V^*\Lambda_iVx=V^*\Lambda_i\begin{bmatrix}
             V_1x\\
             \vdots\\
             V_nx
         \end{bmatrix}=V_i^*V_ix=A_ix, \forall x\in \mathcal{H}_0\overset{\bot}{\mathcal{\oplus}}\mathcal{H}_1
\end{gather}
for each $i=1,\ldots, n$. In particular, this proves that
\begin{gather}
    A(z)=V^*L(z)V\qquad\forall z\in\mathbb{C}^n.
\end{gather}
Now $A(1,\ldots, 1)=\sum_{i=1}^nA_i=I$ (the identity operator on $\mathcal{H}_0\overset{\bot}{\mathcal{\oplus}}\mathcal{H}_1$) by assumption \eqref{NormalizationCond} and $L(1,\ldots,1)=\sum_{i=1}^n\Lambda_i=I$ (the identity operator on $\mathcal{H}$) by construction, hence
\begin{gather}
    V^*V=I.
\end{gather}
From this and Lemma \ref{lem:TUnitaryTransRestrV}, one concludes that one has an orthogonal triple decomposition of $\mathcal{H}$ with 
	\begin{gather}
		\mathcal{H=U}\overset{\bot}{\mathcal{\oplus}}\mathcal{E}\overset{\bot
			}{\mathcal{\oplus}}\mathcal{J},\\
         \mathcal{U}=V\mathcal{H}_0,\; \mathcal{E}=V\mathcal{H}_1,\;\mathcal{J}= \mathcal{H}\overset{\bot}{\mathcal{\ominus}}(\mathcal{U}\overset{\bot}{\mathcal{\oplus}}\mathcal{E}). 
	\end{gather}
By Lemma \ref{lem:TUnitaryTransRestrV}, it follows that the restrictions of the operators $V$ and $V^*$, defined by
        \begin{gather}
            T:\mathcal{H}_0\rightarrow V\mathcal{H}_0, Tv=Vv,\; v\in \mathcal{H}_0,\\ S:V\mathcal{H}_0\rightarrow \mathcal{H}_0,\; Sy=V^*y, \; y\in V\mathcal{H}_0,
        \end{gather}
        respectively, are unitary transformations that are inverses of each other, i.e.,
        \begin{gather}
            T\in \mathcal{L}(\mathcal{H}_0,V\mathcal{H}_0),\;S\in \mathcal{L}(V\mathcal{H}_0,\mathcal{H}_0), \; T^{-1}=T^*=S.
        \end{gather}
It has also been proven that the $Z$-problem
\begin{gather}
    (\mathcal{H},\mathcal{U},\mathcal{E},\mathcal{J},L(z))
\end{gather}
is an abstract $n$-phase composite with orthogonal $Z(n)$-subspace collection
\begin{gather}
    \mathfrak{C} = (\mathcal{H},\mathcal{U},\mathcal{E},\mathcal{J},\mathcal{P}_1,\ldots,\mathcal{P}_n).
\end{gather}
Consider the effective operator $L_*(z)$ of $\mathfrak{C}$. It is claimed that
\begin{gather}
    A(z)/A_{11}(z)=T^*L_*(z)T\;\quad\forall z\in D.
\end{gather}
Let $z\in D$. Then by Proposition \ref{prop:EffOpBessRealization}, the effective operator $L_*(z)$ is well-defined and by the assumptions on $A(z)$ the block $A_{11}(z)$ is invertible so that the Schur complement $f(z)=A(z)/A_{11}(z)$ is well-defined. Next, let $x_0\in \mathcal{H}_0$. Then $w_0=A(z)/A_{11}(z)x_0\in \mathcal{H}_0$ and $x_1=-A_{11}(z)^{-1}A_{10}(z)x_0\in \mathcal{H}_1$ satisfy
\begin{gather}
    w_0=A(z)(x_0+x_1).
\end{gather}
Hence,
\begin{gather}
    Tw_0+0=Vw_0=VA(z)V^*(Vx_0+Vx_1)=L(z)(Tx_0+Vx_1),
\end{gather}
and since $Tx_0\in \mathcal{U}$ with $(Tw_0,Vx_1,0)\in \mathcal{U}\times\mathcal{E}\times\mathcal{J}$, it follows that
\begin{gather}
    Tw_0=L_*(z)Tx_0,
\end{gather}
implying
\begin{gather}
    A(z)/A_{11}(z)x_0=w_0=T^*L_*(z)Tx_0.
\end{gather}
This proves the claim, which completes the proof of the theorem.
\end{proof}

\section{Effective conductivity in a periodic medium}\label{sec:EffCondPeriodicMedium}

Consider the Hilbert space of periodic square-integrable vector-valued functions $\left[  L_{\#}^{2}\left(\Omega\right)\right]^{d}$ ($d=2$ or $d=3$, over the field $\mathbb{K}=\mathbb{R}$ or $\mathbb{K}=\mathbb{C}$) with unit cell 
$\Omega\subseteq \mathbb{R}^n$ [e.g., $\Omega = (0,2\pi)^d$] and the Hodge decomposition 
\begin{gather}
    \left[  L_{\#}^{2}\left(\Omega\right)\right]^{d}=\mathcal{H}=\mathcal{U}\ho\mathcal{E}\ho\mathcal{J},\label{HodgeDecompPeriodContinuumCond}\\
	\mathcal{U}  =\{U\in\mathcal{H}:\langle U \rangle=U\}, \label{HodgeDecompPeriodContinuumCond1}\\
	\mathcal{E}  =\{E\in\mathcal{H}:\nabla\times E=0,\; \langle E\rangle=0\}, \label{HodgeDecompPeriodContinuumCond2}\\
	\mathcal{J}  =\{J\in\mathcal{H}:\nabla\cdot J=0,\; \langle J\rangle=0\},\label{HodgeDecompPeriodContinuumCond3}
\end{gather}
with the inner product and (unit cell) average
\begin{align}
	\left(E,F\right)_{\mathcal{H}}=\frac{1}{\left\vert \Omega\right\vert }%
	{\textstyle\int\limits_{\Omega}}\overline{E\left(x\right)}^{T}F(x)dx,\;\; 
	\left\langle F\right\rangle                 
	=\frac{1}{\left\vert \Omega\right\vert }%
	{\textstyle\int\limits_{\Omega}}            
	F\left(  x\right)  dx,    \label{ExContinuumPeriodicCondZProbRealInnerProdAndAvg}          
\end{align}
respectively, for all $E,F\in\mathcal{H}$. Here, $(\cdot)^T$ and $\overline{(\cdot)}$ denote the transpose and complex conjugation, respectively, and $|\Omega|=\int_{\Omega}dx$ is the Lebesgue measure of $\Omega$. 
In particular, $\mathcal{U}$ is the $d$-dimensional space of uniform (constant) vector functions; $\mathcal{E}$ is the infinite-dimensional space of all $\Omega$-periodic fields $E$ characterized by $E=\nabla u$ for some $\Omega$-periodic function $u$; $\mathcal{U}\ho \mathcal{E}$ and $\mathcal{U}\ho \mathcal{J}$ are the spaces of periodic vector fields which are the gradient $\nabla$ of a potential and divergence-free, respectively. 

The $Z$-problem $(\mathcal{H},\mathcal{U},\mathcal{E},\mathcal{J},L)$ for this example is the problem of determining, for a given $E_0\in \mathcal{U}$, the triplet $J_0\in \mathcal{U}$, $E\in \mathcal{E}$, $J\in \mathcal{J}$ such that the following Ohm's law with conductivity $L$ holds:
\begin{align}
    J_0+J=L (E_0+E).\label{DefZProbContConductivity}
\end{align}
Effective conductivity (operator) $L_*\in \mathcal{L}(\mathcal{U})$ is just an effective operator for this $Z$-problem, i.e., it satisfies
\begin{align}
    J_0=L_*E_0,\label{DefZProbContEffConductivity}
\end{align}
whenever $E_0\in \mathcal{U}$ and $(J_0,E,J)\in \mathcal{U}\times \mathcal{E}\times \mathcal{J}$ solves Ohm's law \eqref{DefZProbContConductivity}. Equivalently, for the periodic electric field $E_0+E\in \mathcal{U}\ho \mathcal{E}$ and the periodic current density $J_0+J\in \mathcal{U}\ho \mathcal{J}$ related through Ohm's law (\ref{DefZProbContConductivity}), the effective conductivity $L_*$ relates the averages of these fields:
\begin{align}
    J_0=\langle J_0+J \rangle = \langle L(E_0+E) \rangle = L_* \langle E_0+E \rangle=L_*E_0.
\end{align}

Consider now the subspace $L_{\#}^{\infty}\left(\Omega\right)\subseteq L_{\#}^{2}\left(\Omega\right)$ of essentially bounded (periodic) functions and denote the set of all $d\times d$ matrices with entries in $L_{\#}^{\infty}\left(\Omega\right)$ by $[L_{\#}^{\infty}\left(\Omega\right)]^{d\times d}$. Then $L_{\#}^{\infty}\left(\Omega\right)$ and $[L_{\#}^{\infty}\left(\Omega\right)]^{d\times d}$ can be treated as subspaces of bounded linear operators in $\mathcal{L}(\mathcal{H})$ when viewed as left-multiplication operators, i.e.,
\begin{gather}
    \sigma\in L_{\#}^{\infty}\left(\Omega\right)\cup [L_{\#}^{\infty}\left(\Omega\right)]^{d\times d}\Rightarrow L_{\sigma}\in \mathcal{L}(\mathcal{H}), \text{ where } L_{\sigma}(F)=\sigma F,\; F\in \mathcal{H}.
\end{gather}
Similarly, by convention $\mathbb{K}\subseteq L_{\#}^{\infty}\left(\Omega\right)$ is the subspace of constant functions so that $\mathcal{U}=\mathbb{K}^d$ and hence the elements of $\mathcal{L}(\mathcal{U})$ are just operators of left multiplication by matrices in $\mathbb{K}^{d\times d}$, i.e., 
\begin{gather}
    \mathcal{L}(\mathcal{U})= \{L_{\sigma}:\sigma\in \mathbb{K}^{d\times d}\}.
\end{gather}
Moreover, as the map $\mathbb{K}^{d\times d}\ni\sigma\mapsto L_{\sigma}\in \mathcal{L}(\mathcal{U})$ is a bijection, for any effective (conductivity) operator $L_*$ of a $Z$-problem  $(\mathcal{H},\mathcal{U},\mathcal{E},\mathcal{J},L)$ there exists a unique $\sigma_*\in \mathbb{K}^{d\times d}$ such that $L_*=L_{\sigma_*}$. The matrix $\sigma_*$ is called the effective conductivity tensor for this $Z$-problem. 

This terminology is motivated by noting that if $\sigma\in L_{\#}^{\infty}\left(\Omega\right)\cup [L_{\#}^{\infty}\left(\Omega\right)]^{d\times d}$ (local conductivity tensors) then with $L=L_{\sigma}$, the effective conductivity $(L_{\sigma})_*=L_{\sigma_*}$ can be identified with the effective conductivity tensor $\sigma_*\in \mathbb{K}^{d\times d}$. In addition, the orthogonal projection of $\mathcal{H}$ onto $\mathcal{U}$, which has been denoted by $\Gamma_0$, is the averaging operator, i.e.,
\begin{gather}
\Gamma_0\in\mathcal{L}(\mathcal{H}),\; \Gamma_0F=\langle F\rangle,\; F\in \mathcal{H},\\  
\Gamma_0^2=\Gamma_0,\; \Gamma_0^*=\Gamma_0,\; \Gamma_0\mathcal{H}=\mathcal{U}.
\end{gather}

Next, we will consider the effective conductivity of an isotropic $n$-phase composite. Suppose that $\Omega$ is decomposed into a disjoint union
\begin{gather}
    \Omega = \Omega_1\cup \cdots \cup \Omega_n,\\
    \Omega_i\cap \Omega_j=\emptyset\;\; (i\not=j)
\end{gather}
for some Lebesgue measurable sets $\Omega_i\subseteq \Omega$ with $|\Omega_i|>0$ and consider the characteristic function
\begin{gather}
    \Omega\ni x\mapsto\chi_i(x)= \left\{\begin{array}{ll}
        1 & \text{if } x\in \Omega_i,\\
        0 & \text{if } x\in \Omega\setminus \Omega_i,
    \end{array}\right.
\end{gather}
which is extended by periodicity to $\mathbb{R}^d$ so that
\begin{gather}
    \chi_i\in L_{\#}^{\infty}\left(\Omega\right),\;\;i=1,\ldots, n.
\end{gather}
Next, define $\mathcal{P}_i$ to be the range of the multiplication operator $L_{\chi_i}$, i.e.,
\begin{gather}
    \mathcal{P}_i=L_{\chi_i}\mathcal{H}=\{\chi_iF:F\in \mathcal{H}\},\;i=1,\ldots, n.\label{CharFuncsDefSubspaces}
\end{gather}
Then an orthogonal $Z(n)$-subspace collection $\mathfrak{C}$ given by
\begin{gather}
    \mathfrak{C} = (\mathcal{H},\mathcal{U},\mathcal{E},\mathcal{J},\mathcal{P}_1,\ldots,\mathcal{P}_n),
\end{gather}
where $\mathcal{H}$ is the Hilbert space \eqref{HodgeDecompPeriodContinuumCond} with subspaces $\mathcal{U},\mathcal{E},\mathcal{J}$ defined by \eqref{HodgeDecompPeriodContinuumCond1}-\eqref{HodgeDecompPeriodContinuumCond3} and subspaces $\mathcal{P}_1,\ldots,\mathcal{P}_n$ defined by \eqref{CharFuncsDefSubspaces}, satisfies
\begin{gather}
    \mathcal{H}=\mathcal{U}\overset{\bot }{\oplus }\mathcal{E}\overset{\bot }{\oplus }\mathcal{J}=\mathcal{P}_1\overset{\bot }{\oplus }\cdots\overset{\bot }{\oplus }\mathcal{P}_n.
\end{gather}
Letting $\Lambda_i$ denote the orthogonal projection of $\mathcal{H}$ onto $\mathcal{P}_i$, one has
\begin{gather}
    \Lambda_i=L_{\chi_i},\;\Lambda_iF=\chi_iF,\; F\in \mathcal{H}, \qquad i=1,\ldots, n.
\end{gather}
In this case, the $Z$-problem 
\begin{gather}
    (\mathcal{H},\mathcal{U},\mathcal{E},\mathcal{J}, L_{\sigma(z)}),
\end{gather}
where
\begin{gather}
    \sigma(z)=z_1\chi_1+\ldots+z_n \chi_n,\\
    L_{\sigma(z)}=z_1\Lambda_1+\ldots+z_n \Lambda_n,\; z=(z_1,\ldots, z_n)\in\mathbb{C}^n
\end{gather}
is an $n$-phase composite with orthogonal $Z(n)$-subspace collection $\mathfrak{C}$ (cf.\ Def.\ \ref{def:SubspCollMultiphasComposite}). For the effective conductivity of this $Z$-problem, i.e., $(L_{\sigma(z)})_*=L_{\sigma_*(z)}$, where $\sigma_*(z):=\sigma(z)_*$ is the effective conductivity tensor, they are well-defined functions for all $z\in D$ by Proposition \ref{prop:EffOpBessRealization}.

The following corollary is the Keller-Dykhne-Mendelson duality for isotropic $n$-phase composites for conductivity, see \cite[Chap.\ 3]{Milton:2022:TOC} and references therein.
\begin{corollary}
    If $d=2$ and $R\in \mathbb{K}^{2\times 2}$ is the matrix of clockwise rotation by $90$ degrees, i.e.,
    \begin{gather}
        R=\begin{bmatrix}
            0 & 1\\
            -1 & 0
        \end{bmatrix},
    \end{gather} 
    then the following statements are true:
    \begin{itemize}
        \item[(i)] The orthogonal $Z(n)$-subspace collection  $\mathfrak{C} = (\mathcal{H},\mathcal{U},\mathcal{E},\mathcal{J},\mathcal{P}_1,\ldots,\mathcal{P}_n)$ has the properties
        \begin{gather}
            R\mathcal{U}=\mathcal{U}, \;R\mathcal{E}=\mathcal{J}, R\!\mathcal{J}=\mathcal{E},\;R\mathcal{P}_i=\mathcal{P}_i,\; i=1,\ldots, n.
        \end{gather}
        \item[(ii)]\;The effective conductivity tensor $\sigma_*(z)$ satisfies
    \begin{gather}
        \sigma_*(z)=R[\sigma_*(z^{-1})]^{-1}R^{-1}\qquad \forall z\in D.
    \end{gather} 
    \end{itemize}
\end{corollary}
\begin{proof}
    First, the matrix $R$ is invertible with $R^*=R^{-1}=-R$. It follows that the left-multiplication operator $L_{R}\in \mathcal{L}(\mathcal{H})$ is invertible with $L_{R}^*=L_{R}^{-1}=-L_{R}$. Now, clearly $L_R\mathcal{U}=R\mathcal{U}\subseteq \mathcal{U}$ and  $L_R\Lambda_i=L_RL_{\chi_i}=L_{\chi_i}L_R=\Lambda_iL_R$, implying $R\mathcal{P}_i=L_R\mathcal{P}_i\subseteq\mathcal{P}_i$, for $i=1,\ldots, n$. Next, it is well-known (\cite[p.\ 105]{Cherkaev:2000:VMS}, \cite[Sec.\ 3.1]{Milton:2022:TOC}) that a two-dimensional curl free field when rotated pointwise by $90$ degrees produces divergence free field and vice versa, which implies $L_R\mathcal{E}=R\mathcal{E}\subseteq \mathcal{U}\overset{\bot }{\oplus }\mathcal{J}$ and $L_R\mathcal{J}=R\mathcal{J}\subseteq \mathcal{U}\overset{\bot }{\oplus }\mathcal{E}$. Next, as any element $F\in \mathcal{E}\cup \mathcal{J}$ satisfies $\langle RF\rangle=R\langle F\rangle=0$ this implies $R\mathcal{E}\subseteq \mathcal{J}$ and $R\mathcal{J}\subseteq \mathcal{E}$. By Corollaries \ref{cor:AbstractKellerDykhneMendelsonDuality} and \ref{cor:AbstractMultiPhaseKellerDykhneMendelsonDuality}, it follows that $$R\mathcal{U}=\mathcal{U}, \;R\mathcal{E}=\mathcal{J}, R\!\mathcal{J}=\mathcal{E},\;R\mathcal{P}_i=\mathcal{P}_i,\; i=1,\ldots, n,$$ and, for each $z\in D$, \begin{gather}
        L_{\sigma^*(z)}=R[L_{\sigma^*(z^{-1})}]^{-1}R^{-1}=RL_{[\sigma^*(z^{-1})]^{-1}}R^{-1}=L_{R[\sigma^*(z^{-1})]^{-1}R^{-1}},
    \end{gather}
    implying $\sigma_*(z)=R[\sigma_*(z^{-1})]^{-1}R^{-1}$. This completes the proof.
\end{proof}

The next result is an immediate consequence of a corollary (Corollary \ref{cor:AbstWienerBounds}) to Theorem \ref{ThmClassicalThomMinPrin}, where the first part provides the classical Wiener bounds for effective conductivity tensor in which the upper and lower bounds are given by the (weighted) arithmetic and harmonic means, see \cite[Chaps.\ 13 \& 22]{Milton:2022:TOC}, \cite[Chap.\ 2, Sec.\ 7]{Milton:2016:ETC}, \cite[Chap.\ 3, Sec.\ 1]{Cherkaev:2000:VMS}. 
\begin{corollary}
    (i) If $z_i>0$ for $i=1,\ldots,n$, then the effective conductivity tensor $\sigma_*(z)$ satisfies the Wiener bounds
\begin{gather}
    \left(\sum_{i=1}^nz_i^{-1}f_i\right)^{-1}I\leq \sigma_*(z)\leq\left(\sum_{i=1}^nz_if_i\right) I,\label{WienerBdsOpVer}
\end{gather}
where $I$ denotes here the $d\times d$ identity matrix of $\mathbb{K}^{d\times d}$ and $f_i$ is the volume fraction of the $i$th phase, i.e.,
\begin{gather*}
    f_i=\langle\chi_i\rangle=\frac{|\Omega_i|}{|\Omega|}=\frac{\int_{\Omega_i}dx}{\int_{\Omega}dx},\;i=1,\ldots, n.
\end{gather*}\\
(ii) If $0<z_i\leq w_i$ for $i=1,\ldots,n$ then
\begin{gather}
    \sigma_*(z)\leq \sigma_*(w).
\end{gather}
\end{corollary}
\begin{proof}
The proof of (i) follows immediately from Corollary \ref{cor:AbstWienerBounds}.(i) by using the fact that $(L_{\sigma(z)})_*=L_{\sigma_*(z)}$ and $\Gamma_0$ is just the averaging operator $\Gamma_0(\cdot)=\langle \cdot \rangle$
so that
\begin{gather*}
    \Gamma_0\Lambda_i\Gamma_0|_{\mathcal{U}}U=\Gamma_0(\chi_iU)=\langle \chi_iU\rangle=f_i\langle U\rangle=f_iU,\; U\in \mathcal{U}, \qquad i=1,\ldots, n.
\end{gather*}
The proof of (ii) follows from Corollary \ref{cor:AbstWienerBounds}.(ii), since the inequality $L_{\sigma_*(z)}\leq L_{\sigma_*(w)}$ between the left-multiplication operators implies $\sigma_*(z)\leq \sigma_*(w)$.
\end{proof}

\section{Conclusion}
This chapter has presented some fundamental concepts and results on the abstract mathematical theory of composites using the Hilbert space framework. The focus has been on the effective operator and its properties from the viewpoint of operator theory on Hilbert spaces. The approach can be used for a broader range of applied problems involving composites and the abstract framework should be further studied as many open problems still remain to be solved. 

One question that naturally arises from this consideration and deserves more attention is the effective operator \textit{realizability problem}: Find necessary and sufficient conditions for a given function $f(z)$ of $n$ real or complex variables $z=(z_1,\ldots, z_n)$ to be unitarily equivalent to an effective operator of an abstract $n$-phase composite having an orthogonal $Z(n)$-subspace collection.

There are two possible approaches to answering the question: (1) via operator theory on Hilbert spaces using the machinery of Schur complements; (2) using the analytic properties $f(z)$ as a function of $z$. Although the latter approach is worthy of consideration, it will not be discussed further as it lies outside the scope of this chapter and the interested reader is recommended to consult the references \cite{Milton:1987:MCI}, \cite{Milton:1987:MCII}, \cite[Chaps.\ 18 \& 29]{Milton:2022:TOC}, \cite{Milton:2016:SAS}, \cite{Cassier:2016:RAF}, and \cite{Milton:2021:SOP} for more details. 

For the first approach, Proposition \ref{prop:EffOpBessRealization} and Theorem \ref{thm:BessNHPSDPCharViaEffOp} provide the complete solution to the realizability problem. In particular, they describe the fundamental relationship between the following two classes of operator-valued functions: (i) effective operators of abstract multiphase composite having orthogonal subspace collections; (ii) Schur complements of normalized homogeneous positive semidefinite operator pencils. The class (i) of functions has been called the \textit{Milton class} (cf.\ \cite[Chap.\ 7]{Stefan:2021:SCA}). The class (ii) is a special subclass of the \textit{Bessmertny\u{\i} class} of all operator-valued functions which can be represented as a Schur complement of a block of homogeneous positive semidefinite operator pencils. For more on the Bessmertny\u{\i} class, see, e.g., \cite{Bessmertnyi:2002:RRM}, \cite{Alpay:2003:RTR}, \cite{Kalyuzhnyui:2004:BCH}, \cite{Ball:2011:MCS}, \cite{Stefan:EBR:2021}.


\textbf{Acknowledgements:} The author is grateful to the Simons Foundation for the support through grant MPS-TSM-00002799 and to the National Science Foundation for support through grant DMS-2410678.

\bibliographystyle{abbrv}
\bibliography{references}
\end{document}